\documentclass[11pt]{article}
\usepackage[utf8]{inputenc}
\usepackage{enumitem}
\usepackage{fullpage}
\usepackage[skins,xparse,breakable]{tcolorbox}
\usepackage{amsmath,amsfonts,amssymb}
\usepackage{bbm}
\usepackage{circuitikz}
\usepackage{hyperref}
\usepackage{tikz}
\usepackage{standard}
\usepackage{algorithm, algpseudocode}
\usepackage{thm-restate}
\usepackage{dsfont}
\usepackage{amsthm}
\newcommand{\modify}[1]{{#1}}
\usepackage{sectsty}
\usepackage{gastex}
\usepackage{natbib}

\usepackage{authblk}

\newcommand{\wtilde}{\widetilde}

\DeclareMathOperator{\dd}{d} 

\DeclareMathOperator*{\smax}{smax}
\DeclareMathOperator*{\argmax}{arg\,max}

\DeclareMathOperator*{\argsmax}{arg\,smax}

\DeclareMathOperator{\Ber}{Ber}
\newcommand{\dtv}[2]{d_{\mathrm{TV}}\left( {#1}, {#2} \right)}
\newcommand{\KL}[2]{KL\left( {#1} \parallel {#2} \right)}

\newcommand{\eps}{\varepsilon}

\newcommand{\bI}{\mathbb{I}}
\newcommand{\bR}{\mathbb{R}}
\newcommand{\cA}{\mathcal{A}}
\newcommand{\cD}{\mathcal{D}}
\newcommand{\cE}{\mathcal{E}}
\newcommand{\cI}{\mathcal{I}}

\newcommand{\ind}[1]{\mathbb{I}\left\{ {#1} \right\}}
\newcommand{\OPT}{\textsc{Opt}}

\newcommand{\Rev}{\textsc{Rev}}
\newcommand{\unif}{\textsc{Unif}}

\usepackage{braket}

\usepackage{pgfplots}
 
\pgfplotsset{compat = newest}
\begin{document}

\title{Improved Online Learning Algorithms for CTR Prediction in Ad Auctions}
\author[1]{Zhe Feng\footnote{zhef@google.com}}
\author[1]{Christopher Liaw \footnote{cvliaw@google.com}}
\author[2]{Zixin Zhou \footnote{jackzhou@stanford.edu}}
\affil[1]{Google Research}
\affil[2]{Stanford University}

\maketitle

\begin{abstract}
In this work, we investigate the online learning problem of revenue maximization in ad auctions, where the seller needs to learn the click-through rates (CTRs) of each ad candidate and charge the price of the winner through a \emph{pay-per-click} manner.
We focus on two models of the advertisers' strategic behaviors.
First, we assume that the advertiser is completely myopic; i.e.~in each round, they aim to maximize their utility only for the current round.
In this setting, we develop an online mechanism based on upper-confidence bounds that achieves a tight $O(\sqrt{T})$ regret in the worst-case and \emph{negative} regret when the values are static across all the auctions and there is a gap between the highest expected value (i.e.~value multiplied by their CTR) and second highest expected value ad.
Next, we assume that the advertiser is non-myopic and cares about their long term utility. This setting is much more complex since an advertiser is incentivized to influence the mechanism by bidding strategically in earlier rounds.
In this setting, we provide an algorithm to achieve \emph{negative} regret for the static valuation setting (with a positive gap), which is in sharp contrast with the prior work that shows $O(T^{2/3})$ regret when the valuation is generated by adversary.

\end{abstract}

\section{Introduction}\label{sec:intro}
Pay-per-click auctions are widely used in internet advertising auctions to allocate advertising space to advertisers~\citep{edelman2007internet, Varian07}.
As a concrete example, major search engines, such as Bing and Google, run an auction for every search query to decide which ads to show.
An important aspect of these auctions is that ads are only charged in the event of a click. 

A crucial piece of information that is required to run these pay-per-click ad auctions is the \emph{click-through rate (CTR)}, which is the probability that an ad which is shown is actually clicked.
The canonical Vickrey-Clarke-Groves (VCG) auction assigns a score equal to the product of an ad's bid and their CTR.
The winning ad is the ad with the highest score and their payment, in the event of a click, is equal to the second-highest score divided by their own CTR~\citep{aggarwal2006truthful}.

In this paper, we consider the problem where the CTR is not known to the seller and must be estimated from data.
We model this as an online learning problem where at each time $t$, each advertiser $i$ has a private $v_{i, t}$ for their ad to be clicked at time $t$ and places a bid $b_{i, t}$.
The CTR for ad $i$ is $\rho_i$ and we assume this remains static over time. \modify{For presentation simplicity, we call $\rho_i \cdot v_{i, t}$ the expected cost-per-impression (eCPM)}.
The auctioneer must then run an auction at time $t$ (that may depend on all information up to time $t-1$) which receives the bids and outputs a winner and the price in the event of a click.
Our goal is to minimize the auctioneer's regret of not knowing the CTR beforehand;
this is equal to the difference of revenue that they could have achieved using a VCG auction if they knew the CTR and the revenue they obtained using an online mechanism.

\subsection{Our Results}
Our results revolve around two models of advertiser utility: the myopic setting and the non-myopic setting.

\paragraph{Myopic setting.}
We first assume that the advertisers are myopic in the sense that, at a particular round $t$, each advertiser bids to maximize their utility at round $t$.
In particular, the advertiser does not try to use their bid at a particular round to influence future auctions.
In this setting, we design a ``stage-wise incentive compatible (stage-IC)'' auction which combines the upper confidence bound (UCB) algorithm for bandits \citep{AuerCF02} with the canonical VCG auction described above.
By stage-IC, we mean that a myopic bidder is incentived to bid truthfully at each time (a formal definition can be found in Definition~\ref{def:stage-ic}).
This UCB-style auction computes an upper confidence bound on the CTR of each ad and uses the product of this estimate of the CTR with the advertiser's bid as the advertiser score.
As a warmup result, we prove that this simple algorithm achieves $\wtilde{O}(\sqrt{T})$ regret in the worst case where \modify{the values are generated from an adversary}.
Here, $T$ corresponds to the number of rounds and $\wtilde{O}$ hides logarithmic factors in $T$.
We complement this by proving a $\Omega(\sqrt{T})$ lower bound even when the values of the ads remain static across all rounds.

Next, we consider the setting where the values are static across all rounds and there exists a gap between the highest eCPM ad and all other ads.
Specifically, we assume there is some ad $i$ for which $v_i \rho_i > v_j \rho_j$ for all other ads $j \neq i$ and the gap is a \emph{time-indepdent} positive constant. 
Our main technical result is that, using exactly the same UCB algorithm as above, the regret is $-\Omega(T)$.

To summarize, we have the following informal theorem.
The formal statements can be found in Theorem~\ref{thm:sqrt_ub}, Theorem~\ref{thm:neg_ub}, and Theorem~\ref{thm:lb} respectively.
\begin{theorem}[Informal]
If the advertisers are myopic then there is an online algorithm that guarantees the following.
\begin{enumerate}[topsep=0pt, itemsep=0pt]
\item The worst-case regret is $\wtilde{O}(\sqrt{T})$.
\item If \modify{the values are static across all rounds and} there is a time-independent constant gap between the highest eCPM ad and all other ads then the regret is $-\Omega(T)$.
\end{enumerate}
\end{theorem}
In Subsection~\ref{subsec:techniques} we give a high-level description of the techniques used in the proof.

\paragraph{Non-myopic setting.}
In the second setting, we assume that the advertisers are non-myopic.
In particular, we assume that the advertisers want to maximize their total utility over all $T$ rounds.
In this case, we design a ``global-IC'' mechanism that also achieves $-\Omega(T)$ regret provided that the values are static and there is a gap between the highest value ad and all other ads (see Theorem~\ref{thm:static_non_myopic} for the formal result).
Here, by global-IC we mean that the mechanism must incentivze a non-myopic advertiser to bid truthfully at every round (see Definition~\ref{defn:global_ic} for a formal definition).
In this setting, the algorithm we design is based on both the UCB algorithm and the explore-then-commit algorithm from the bandit literature.
In particular, the algorithm runs in two phases.
In the first phase, the algorithm explores by showing each ad for free in a round-robin manner.
At the end of each round-robin, we compute confidence intervals for all the ads.
We terminate the phase if there is an ad whose lower confidence bound is greater than the upper confidence bound of all other ads.
Note that the termination time for the first phase is not provided as input to the algorithm.
In the second phase, the algorithm commits by using a UCB estimate of the CTR of each of the ads.

\subsection{Techniques}
\label{subsec:techniques}
There are a number of key differences between our setting and the standard multi-armed bandit setting which calls for novel ideas.
First, in the multi-armed bandit setting, there is a single fixed arm that the algorithm needs to perform well against.
In our setting, the benchmark is the revenue that can be achieved if the algorithm had known the CTRs beforehand.
At time $t$ this revenue is equal to second highest quantity in the set $\{v_{i, t} \rho_i\}$ where $v_{i,t}$ is the value of ad $i$ and $\rho_i$ is the CTR of ad $i$.
Note, importantly, that the benchmark may actually be \emph{different} at each time step.

A second key difference is that the algorithm may actually incur regret even when it chooses the correct ad to be shown.
As an example, suppose that we have two ads with CTR $\rho_1 = \rho_2 = 0.5$.
Ad $1$ has value $1$ while ad $2$ has value $2$.
Our estimates for their CTRs are $\hat{\rho}_1 = 0.5$ and $\hat{\rho}_2 = 0.8$.
If we had known $\rho_1$ and $\rho_2$ then the canonical VCG auction scores each ad by $v_i \rho_i$.\footnote{This assumes the advertisers bid truthfully but since the auction is a standard VCG auction, which is stage-wise IC, it is reasonable to assume truthful bidding.}
The winner is ad $2$ (since $v_2 \rho_2 > v_1 \rho_1)$ and, in the event of a click, ad $2$ is charged $v_1 \rho_1 / \rho_2$ (the losing score divided by the winner's CTR).
The expected revenue is thus $v_1 \rho_1 = 0.5$.
On the other hand, suppose that one uses the the estimated CTRs.
Then ad $2$ still wins but is charged $v_1 \hat{\rho}_1 / \hat{\rho}_2$ when they are clicked for an expected revenue of $v_1 \hat{\rho}_1 \cdot \frac{\rho_2}{\hat{\rho}_2} = 0.3125$.
Thus, we still incur regret even though the correct ad is shown.
The reason that we are still incurring regret in this example is that our estimate for the CTR of ad $2$ is incorrect.
However, observe that the closer that $\hat{\rho}_2$ is to $\rho_2$ (while still being larger), the smaller the incurred regret when the correct ad is shown.

To prove our upper bound result, we first observe that the regret can be expressed as the sum of two components.
The first component is the regret that is incurred whenever we choose an incorrect ad at each time step.
While the correct ad may actually be different at each time step, we are still able to bound the regret of this component using a fairly standard argument.
The second component is the regret that is incurred when we show the correct ad but our estimate of the CTR is not accurate.
Intuitively, if we have shown the correct ad $K$ times then this estimation error should be roughly on the order of $1/\sqrt{K}$.
We show that this also translates to an incurred regret of $1/\sqrt{K}$.
It turns out that this blueprint is sufficient to prove the $\wtilde{O}(\sqrt{T})$ upper bound.

Proving the negative regret bound for static values calls for additional ideas.
As in the above paragraph, we still split the regret into a component for showing incorrect ads and a component for showing correct ads.
The former can still be bounded using a fairly standard algorithm.
Our goal then is to show that the regret due to showing correct ads is actually \emph{negative}.
This requires opening up the blackbox of the UCB algorithm.
Recall that in this setting, we are also assuming that there is a gap between the best ad to show and all the other ads.
Specifically, sorting the ads, we assume that $\rho_1 v_1 > \rho_2 v_2 \geq \ldots \geq \rho_n v_n$.
In this case, we need to show that UCB does \emph{not} maintain a very good estimate of $\rho_2$.
This is in stark contrast with the usual proofs involving UCB that only require proving that the UCB estimate is good in order to prove regret bounds.
Note that the UCB estimate is always an upper bound on $\rho_2$; the key observation is that, since there is a gap, this estimate must actually be \emph{significantly} more than $\rho_2$.
As a result, this allows to conclude that ad $1$ is actually seeing slightly more competition which results in the algorithm setting a higher price per click.

\subsection{Related Work}\label{subsec:related-work}
The closely related work lies in the field of \emph{online learning for pay-per-click auctions}~\citep{DK09, BSS14}. \citet{DK09} show $\Theta(T^{2/3})$ regret against the revenue achieved by second price auctions when the true click-through rates are known. They restrict their attention to the global-IC mechanism, i.e., bidding truthfully is the dominant strategy for each advertiser given any realized sequence of clicks and bids in $T$ rounds. \citet{BSS14} focus on maximizing welfare and propose a black-box reduction from a standard multi-arm bandit (MAB) algorithm to a global-IC MAB algorithm and they show a $\Omega(T^{2/3})$ is necessary for any deterministic global-IC mechanism. \citet{BKS15} further extend to randomized global-IC mechanism and achieve an improved regret bound $O(\sqrt{T})$ when the valuation of advertisers is generated stochastically. In the myopic setting, our UCB-style mechanism is deterministic and stage-IC. We can achieve $\Theta(\sqrt{T})$ (matching lower bound) regret when the valuation of the advertisers are chosen adversarially. Furthermore, we show this algorithm can achieve better revenue than second price auctions with known true click-through rates (i.e., negative regret) when the value of advertisers are fixed, which is surprising. In the non-myopic setting, the explore-then-commit UCB mechanism is global-IC and we show it also achieves \emph{negative} regret, which complements the $O(T^{2/3})$ regret analysis proposed by~\citep{DK09} when the valuation is generated adversarially.

Our work is also related with the Bayesian Incentive Compatible (BIC) Bandits literature in general, e.g.,~\citep{bic-bandit, SS21} and some follow-up works for combinatorial bandits~\citep{bic-combinatorial-bandit} and reinforcement learning~\citep{bic-rl}, where the previous papers focus on designing truthful (global-IC) mechanism to incentivize agents to explore. The online mechanism designed in this paper is stage-IC (i.e., incentive compatible at each round), so that the myopic advertisers will report their true value, even though designing truthful mechanism is not the target of this paper. Our work differs with this BIC bandits literature as: (1) existing BIC bandit papers focus on welfare maximization (i.e., maximize total reward of pulled arms), while our objective is to maximize revenue, which also depends on the 
arms that are not pulled; (2) standard BIC bandits are tailored to stochastic setting, which assume there is a prior belief of each arm, whereas, our UCB-style mechanism achieves $O(\sqrt{T})$ regret when the advertisers' value are generated adversarially. Furthermore, we achieve \emph{negative} regret when the values are fixed across rounds.
This requires careful analysis tailored to pay-per-click second price auctions.

Loosely related work includes the rich literature about learning click-through rates, e.g.,~\citep{mcmahan2013ad,chen2016deep,cheng2016wide,zhang2016deep,qu2016product,juan2017field,lian2018xdeepfm,zhou2018deep}. These existing works focus on the offline setting, which treat predicting CTR as a standard classification question and deep neural networks have proven to be very powerful for this task. Our work focus on the online learing setting where the advertisers and the seller interact at each round, and the seller's target is to learn the click-through rates to maximize the expected revenue.
\section{Model and Preliminaries}
In this section, we describe the model considered in this paper.
We assume a repeated single-slot ad auction setting with $T$ rounds and $n$ advertisers (equivalently, ads) per round. From now on, we will interchangeably use arms and advertisers (or ads), when the context is clear. For simplicity, we assume that $T$ is fixed and known beforehand; our results can be extended to the setting where the number of rounds is not known using a standard doubling trick argument~\citep{BCB12}.
Each ad $i$ has an unknown click-through-rate $\rho_i \in (0, 1)$ , which is fixed across all rounds.
At each round $t$, ad $i$ has a private value $v_{i, t} > 0$ for being clicked; its expected value for being \emph{shown} is thus $v_{i,t} \cdot \rho_i$, which is also usually called expected cost-per-impression (eCPM).
In this paper, we use $\smax$ to denote the set-valued function that returns the largest element in the set (if there are duplicates then $\smax$ returns the largest element).
Similarly, we use $\argsmax$ to denote the set-valued function that returns the index of the second largest element.

Let $b_t = \{b_{1, t}, b_{2, t}, \cdots, b_{n, t}\}$ be the bid profile at round $t$, where $b_{i, t}$ is the bid of advertiser $i$ at this round. Following the common notations of auction theory literature, we denote $b_{-i, t}$ and $v_{-i, t}$ to be the bids and value of the other advertisers except for $i$. 
For each advertiser $i$, the seller specifies an auction which is determined by an allocation rule (can be randomized) $\hat{x}_{i, t} \colon \bR_{\geq 0}^n \to [0, 1]$ where $\sum_{i \in [n]} \hat{x}_{t, i}(b_t) \leq 1$ and an expected payment rule $p_{i, t} \colon \bR_{\geq 0}^n \to \bR_{\geq 0}$.\footnote{Given the expected payment $p_{i,t}$, the payment that is charged in the event of a click is $p_{i,t} / (\rho_i \hat{x}_{i,t}) = p_{i,t} / x_{i,t}$.}  We define $x_{i, t} = \rho_i \hat{x}_{i, t}$ as the effective click probability of ad $i$ at round $t$ and it represents the effective allocation probability \modify{that advertiser $i$ can get the value
in the auction}.
The expected utility of ad $i$ with value $v_{i, t}$ for submitting a bid $b_{i, t}$ conditioned on the remaining bids of the other advertisers are $b_{-i, t}$ at round $t$ is given by 
\begin{eqnarray}\label{eq:utility-per-round}
u_{t, i}(v_{i,t}; b_{i, t}, b_{-i, t}) = v_{i, t} x_{i, t}(b_t) - p_{i, t}(b_t).
\end{eqnarray}
Note that to implement this auction in practice, we only need to define the allocation rule $\hat{x}$ and per-click payment rule. 

In this paper, we consider two different types of advertisers: \emph{myopic} advertisers and \emph{non-myopic} advertisers. In the myopic setting, each advertiser is only interested in maximizing her utility at each round and ignores the effects in the future rounds, i.e., advertiser $i$ would like to submit a bid $b_{i, t}$ to maximize her utility at each round defined in Eq.\eqref{eq:utility-per-round}, conditioned on the other advertisers' bids $b_{-i, t}$. An online mechanism that always incentives the myopic advertisers to report their true value at each round is called \emph{stage-wise incentive compatible} (stage-IC), defined in the following:

\begin{definition}[Stage-IC]\label{def:stage-ic}
An online mechanism is said to be stage-IC if for every $i, t, v_{i,t}, b_{i,t}, b_{-i, t}$, we have
\[
u_{i,t}(v_{i, t}; v_{i,t}, b_{-i, t})
\geq u_{i,t}(v_{i,t}; b_{i,t}, b_{-i, t}).
\]
\end{definition}
In other words, reporting truthfully is a dominant strategy for each myopic advertiser $i$ no matter what the other advertisers submit at each round. The seminal work by~\citet{myerson1981optimal} characterizes the expected payment rule $p_{i, t}$ for stage-IC mechanisms:

\begin{lemma}[\citet{myerson1981optimal}]\label{lem:myerson}
An online mechanism is stage-IC if and only if the allocation rule $x_{i, t}$ is monotone with respect to $i$th coordinate (advertiser $i$'s bid) and the payment rule is given by
\begin{equation}
p_{i, t}(b_t) = b_{i, t} \cdot x_{i, t}(b_t) - \int_0^b x_{i, t}(z, b_{-i, t})\, \dd z.
\end{equation}
for any $i, t, b_{i, t}$ and $b_{-i, t}$.
\end{lemma}

As an illustrative example, we define the pay-per-click second price auctions as follows:
\begin{definition}[Pay-per-click Second Price Auctions]\label{def:spa-ppc}
At each round $t$, the pay-per-click second price auction
\begin{itemize}[topsep=0pt, itemsep=0pt]
\item selects the ad with the highest estimated eCPM score, $i^* = \argmax_i f_{i, t} b_{i, t}$,
where $f_{i, t}$ is an estimate of CTR $\rho_i$; and
\item shows ad $i^*$. If it is clicked, charge $i^*$ by $\frac{\max_{j\neq i} f_{j, t} b_{j, t}}{f_{i, t}}$.
\end{itemize}
\end{definition}
It is well-known, this mechanism is stage-IC, as long as CTR estimates $f_{i, t}$ are independent of $b_{i, t}$~\citep{aggarwal2006truthful}.

In the \emph{non-myopic} setting, each advertiser $i$ aims to maximize her cumulative utility achieved in $T$ rounds, i.e., $\sum_{t=1}^T u_{i, t}(v_{i, t}; b_{i, t}, b_{-i, t})$, by submitting $b_{i, t}$ at each round $t$.  In this setting, advertisers are incentived to leverage
their private information and “game" the system by submitting corrupted bids (deviating from her true value) stratigically in the earlier rounds so that she can manipulate her predicted click-through rates to benefit in the long run. To incentivize non-myopic advertisers to reveal their true valuation, we need a stronger notion of incentive compatibility to design online mehcanisms, which is global-IC.

\begin{definition}[Global-IC]
\label{defn:global_ic}
An online mechanism is said to be global-IC if for every $i, v_{i,t}, b_{i,t}, b_{-i, t}$, we have
\[
\sum_{t=1}^T u_{i,t}(v_{i, t}; v_{i,t}, b_{-i, t})
\geq
\sum_{t=1}^T u_{i,t}(v_{i,t}; b_{i,t}, b_{-i, t}).
\]
\end{definition}

Beyond incentive properties mentioned above, the online mechanism considered in this paper should satisfy individual rationality, i.e., the expected utility of each advertiser is non-negative.
\begin{definition}[Individual Rationality (IR)]\label{def:ir}
An online mechanims is IR if and only if for every $i, t, v_{i, t}, b_{-i, t}$,
\begin{eqnarray}
u_{i,t}(v_{i,t}; v_{i,t}, b_{-i, t}) \geq 0
\end{eqnarray}
\end{definition}

\paragraph{Valuation Generation.} 
In this paper we mainly focus on two important valuation generation settings: (1) valuation of the advertisers are generated from an \emph{adversary}; (2) advertisers' value are \emph{fixed} and \emph{static} across all rounds, i.e, for all $i, t, v_{i,t} = v_i$. Conceptually, the adversarial valuation generation focus on the worst-case scenario and the fixed valuation setting is an ideal and theory-driven model. However, the fixed valuation setting also captures the practice to some extent, e.g., the valuation of a shopping ad is just the price of the product, which is static and fixed in a long period.

\paragraph{Regret.} In this paper, we would like to design an online mechanism to minimize the regret against the revenue of second price auctions (VCG) when the true click-through rates are known. Let $\OPT = \sum_{t=1}^T \smax_{i \in [n]} \rho_{i} v_{i, t}$ be the revenue of second price auctions if the true CTRs $\{\rho_i\}_{i=1}^n$ are known, then the regret is defined as follows.
\begin{definition}[Regret]\label{def:regret}
For an IC mechanism (either stage-IC or global-IC) $\cA$, let $\Rev_T(\cA)$ be the revenue of $\cA$ if all advertisers bid their true value for all rounds.
We define the regret as
\[
    R_T = \OPT - \Rev_T(\cA).
\]
\end{definition}
In this paper, we focus on the revenue acheived by VCG as the benchmark for regret bound, due to the robustness of the VCG in practice, i.e., VCG mechanism does not need to know the prior distribution of the valuation.
Note that $\OPT$ is the optimal revenue that one can achieve if the CTRs are known and if, at every round, we always decide to show an ad. 
\section{A UCB-style Mechanism for Myopic Advertisers}
\label{sec:ucb_no_regret}
In this section, we present a UCB-style pay-per-click online mechanism to minimize the regret (Definition~\ref{def:regret}), when the advertisers are myopic.
The algorithm is built upon standard Upper Confidence Bound (UCB) bandit algorithm. Namely, we maintain a UCB estimate of each ad's pCTR at each round $t$:
\begin{eqnarray}\label{eq:ucb-pctr}
\tilde{\rho}_{i, t} = \hat{\rho}_{i, t} + \sqrt{\frac{3 \log T}{2N_{i, t}}},
\end{eqnarray}
where $N_{i, t}$ is the total number of showing ad $i$ up to time $t-1$ and $\hat{\rho}_{i, t}$ is the average clicks among $N_{i, t}$ ad impressions. We will then run pay-per-click second price auctions (see Definition~\ref{def:spa-ppc}) using the UCB estimates of the CTRs and we defer the details of this UCB-style online mechanism to Algorithm~\ref{alg:ucb_auction}. Following a standard regime in UCB algorithm, we run a forced exploration for each arm $i$ in the beginning to get a warm start for the main UCB online mechanism (Line~\ref{line:warm-start}). Since the bids are not used in this inital exploration, it has no effect on the incentive property of the online mechanism. In our regret analysis, we ignore the regret suffered due to this initial exploration as it can be easily bounded by a constant.

\begin{algorithm}
\caption{UCB-style algorithm for online pay-per-click auctions}
\label{alg:ucb_auction}
\begin{algorithmic}[1]
\State Show each ad $i \in [n]$ once (for free) and observe click; initialize $\hat{\rho}_{i, 1} = \ind{\text{ad $i$ was clicked}}$ and $N_{i, 1} = 1$.\label{line:warm-start}
\For{$t = 1, \ldots, T$}
\State Compute $\tilde{\rho}_{i, t} = \hat{\rho}_{i, t} + \sqrt{\frac{3 \log T}{2N_{i, t}}}$. \label{line:ucb_score}
\State Solicit bids $b_{i, t}$ for each advertiser $i \in [n]$.
\State Let $A_t \in \argmax_{i \in [n]} \tilde{\rho}_{i, t} \cdot b_{i, t}$ (winner) and $B_t \in \argsmax_{i \in [n]} \tilde{\rho}_{i, t} \cdot b_{i, t}$ (runner up). \label{line:ucb_ranking}
\State Show ad $A_t$. Let $X_t = \ind{\text{ad $A_t$ was clicked}}$ and charge $\frac{\tilde{\rho}_{B_t, t} \cdot b_{B_t, t}}{\tilde{\rho}_{A_t, t}} \cdot X_t$ to ad $A_t$ (other ads pay $0$). \label{line:ucb_payment}
\State Update $N_{A_t, t+1} = N_{A_t, t} + 1$ and $N_{i, t+1} = N_{i, t}$ for $i \neq A_t$.
\State Update $\hat{\rho}_{A_t, t+1} = \left( 1 - \frac{1}{N_{A_t, t+1}} \right)\hat{\rho}_{A_t, t} + \frac{1}{N_{A_t, t+1}} X_t$ and $\hat{\rho}_{i,t+1} = \hat{\rho}_{i,t}$ for $i \neq A_t$.
\EndFor
\end{algorithmic}
\end{algorithm}

Our UCB estimates of CTR $\tilde{\rho}_{i, t}$ is indepdent with its bid $b_{i, t}$ for each advertiser $i$ at round $t$. Therefore, our UCB-style online mechanism is clearly stage-IC by~\citep{aggarwal2006truthful},

\begin{restatable}{proposition}{ucbisic}
\label{prop:myerson1}
The UCB-style online mechanism proposed in Algorithm~\ref{alg:ucb_auction} is stage-IC.
\end{restatable}
For completeness of the context, we provide a proof for this Proposition in Appendix~\ref{app:ucb-ic}. Since the UCB-style online mechanism is stage-IC, the myopic advertisers will report their true value at each round.

\subsection{Adversarial Valuation}\label{subsec:adversary-value}
We consider the valuation of the advertisers are generated from adversary in this subsection. 
Intuitively, standard UCB algorithm is tailored to stochastic bandit and cannot work for the adversarial setting. Our result doesn't contradict this common sense because the valuation of all advertisers can be observed in the beginning of each round, due to the truthfulness of our mechanism. In other words, the adversarial valuation can be treated as an adversarially generated but known context and the uncerntainty of this problem is only from the CTRs.

We show in Theorem~\ref{thm:sqrt_ub} that our online mechanism proposed in Algorithm~\ref{alg:ucb_auction} can achieve $\widetilde{O}(\sqrt{T})$ regret and the complete proof is deferred to Appendix~\ref{app:sqrt_ucb}.
\begin{restatable}{theorem}{sqrtub}
\label{thm:sqrt_ub}
Let $M$ be a positive constant s.t. $M \geq \smax_{i, t} \rho_{i} v_{i, t}$.
Then the regret achieved by Algorithm~\ref{alg:ucb_auction} for the adversarial valuation setting can be bounded by,
\[
R_T \leq M \cdot \sum_{i=1}^n \frac{\sqrt{24 T \log(2nT)}}{\rho_i} + \frac{M}{T}.
\]
\end{restatable}

\subsection{Fixed Valuation}\label{subsec:fixed-value}

Recall the fixed valuation setting that the values are fixed and static over time, i.e., $v_{i, t} = v_i, \forall i, t$. In this setting, our UCB-style online mechanism achieves \emph{negative} regret, as long as there is a \emph{time-independent} constant gap between the optimal winner (highest eCPM) and runner up (second highest eCPM), i.e.
\begin{eqnarray}\label{eq:min-gap}
\zeta := \max_{i \in [n]} 
\rho_i v_i - \smax_{i \in [n]} \rho_i v_i > 0
\end{eqnarray}
such that $\zeta > 0$ is a \emph{time-independent} constant. Without this \emph{time-independent} constant gap assumption, we show $\Omega(\sqrt{T})$ lower bound for the fixed valuation setting in Section~\ref{sec:lower-bound}.

Furthermore, the \emph{negative} regret achieved by our algorithm  is \emph{linear} in $T$, which implies that our online mechanism can achieve (unbounded) revenue gain that scales linearly with the number of total rounds, compared with second price auctions associated with true CTRs. 
To simplify notations, it is without loss of generality to assume arm 1 is the optimal arm that has the highst score of eCPM $\rho_i v_i$ and denote
\begin{eqnarray}\label{eq:gap-per-arm}
\forall i = 2, 3,\cdots, n, \Delta_i := \frac{\rho_1 v_1 - \rho_i v_i}{v_i}
\end{eqnarray}

\begin{theorem}
\label{thm:neg_ub}
In the fixed valuation setting,
the UCB-style online mechanism proposed by Algorithm~\ref{alg:ucb_auction} obtains a regret bounded by,
\begin{align*}
R_T
& \leq -0.05 \zeta  T + O(\log(nT))
\end{align*}
where $\zeta > 0$ is a \emph{time-independent} constant defined in Eq.\eqref{eq:min-gap}.
\end{theorem}
\begin{proof}[Proof Sketch]
Let $s = \arg \mathrm{smax}_{i} \, \rho_i\cdot v_i$.
As defined in Algorithm~\ref{alg:ucb_auction}, let $\hat{\rho}_{i, t}$ be the empirical mean of $\rho_i$ at the time $t$ and
$\tilde{\rho}_{i, t} = \hat{\rho}_{i,t} + \sqrt{\frac{3 \log T}{2N_{i,t}}}$ is the UCB estimate of $\rho_i$ at time $t$.
Let $A_t = \arg \mathrm{max}_{i} \, \hat{\rho}_{i,t}\cdot v_i$ be the index of the winning ad in round $t$ and $B_t = \arg \mathrm{smax}_{i} \, \hat{\rho}_{i,t}\cdot v_i$ be the index of the runner up.
Finally, let $\Delta_i = \min\left\{1, \frac{\rho_1 v_1 - \rho_i v_i}{v_i} \right\}$.
Let $X_t = \ind{\text{ad $A_t$ was clicked}}$.
First, observe that we can write the regret as
\begin{align}
&R_T\\
&= \E \left [\sum_{t = 1}^T \sum_{i, j=1}^n  \left ( \rho_s v_s - \tilde{\rho}_{j,t} v_j \frac{X_t}{\tilde{\rho}_{i,t}}\right) \mathbb{I} \left \{ A_t = i, B_t = j\right \}  \right] \notag \\
& = \E \left [\sum_{t = 1}^T \sum_{i = 2, j=1}^n  \left ( \rho_s v_s - \tilde{\rho}_{j,t} v_j \frac{X_t}{\tilde{\rho}_{i, t}} \right) \mathbb{I} \left \{ A_t = i, B_t = j\right \}  \right] \label{eqn:wrong_winner} \\
& + \E \left [\sum_{t = 1}^T \sum_{j=1}^n \left ( \rho_s v_s - \tilde{\rho}_{j,t} v_j\frac{X_t}{\tilde{\rho}_{1,t}} \right ) \mathbb{I} \left \{ A_t = 1, B_t = j\right \}  \right]. \label{eqn:right_winner}
\end{align}
Eq.~\eqref{eqn:wrong_winner} corresponds to rounds where we choose a suboptimal advertiser as the winner and Eq.~\eqref{eqn:right_winner} corresponds to rounds where we choose the optimal advertiser as the winner.
To complete the proof of this theorem, we bound Eq.~\eqref{eqn:wrong_winner} and Eq.~\eqref{eqn:right_winner} in Lemma~\ref{lem:wrong_winner} and Lemma~\ref{lem:right_winner} respectively.

\begin{restatable}{lemma}{wrongwinner}
\label{lem:wrong_winner}

$ \\ \\  \E \left [\sum_{t, i, j} \left ( \rho_s v_s - \tilde{\rho}_{j,t} v_j \cdot \frac{X_t}{\tilde{\rho}_{i, t}} \right) \mathbb{I} \left \{ A_t = i, B_t = j\right \}  \right] \\  \leq \frac{2}{T} + \sum_{i=2}^n \frac{12 \rho_s v_s}{\rho_1} \frac{\log(2nT)}{\Delta_i} + \frac{\rho_s v_s}{\rho_i} \frac{\sqrt{6 \log(2nT)}}{n^2T}$.
\end{restatable}

\begin{restatable}{lemma}{rightwinner}
\label{lem:right_winner}
\begin{align*}
    \E & \left [\sum_{t = 1}^T \sum_{j=1}^n \left ( \rho_s v_s - \tilde{\rho}_{j,t} v_j \frac{X_t}{\tilde{\rho}_{1,t}} \right ) \mathbb{I} \left \{ A_t = 1, B_t = j\right \}  \right] \\
    & \leq -0.05 \Delta_s v_s T + 3\rho_s v_s + \frac{9000 \rho_s^3 v_s \log(2nT)}{\rho_1^s \Delta_s^2} \\
    & ~~ + \frac{450 \rho_s^2 v_s \log(2nT)}{\rho_1^2 \Delta_s} + \frac{0.25 \Delta_s v_s}{nT} \\
    & ~~~~+ \sum_{i=2}^n \frac{0.06 \Delta_s v_s \log(2nT)}{\Delta_i^2} \\
    & \leq -0.05 (\rho_1 v_1 - \rho_s v_s) T + O(\log(nT)).
\end{align*}
\end{restatable}
\end{proof}

The proofs of these two auxilliary lemmas are rather technical and we deferred them to Appendix~\ref{app:wrong_winner_proof} and~\ref{app:right_winner_proof}.
At a high-level, the proof of Lemma~\ref{lem:wrong_winner} bears resemblance to the usual proofs in the bandit literature since this is the regret incurred by showing an incorrect ad. In particular, we first show that the UCB estimate is fairly good with very high problem. Conditioned on this event, we are able to show that the regret is small.
The proof of Lemma~\ref{lem:right_winner} is more technically interesting.
A key observation is that the UCB estimate actually overestimates the true CTR by a good margin (see Lemma~\ref{claim:Delta_gap}).
This observation means that, when we do show the correct ad, the price that the ad is charged is slightly higher due additional competition from an underexplored arm.
For the rigorous details of this argument, we invite the reader to refer to Claim~\ref{claim:rightwinner_1} and its proof in Appendix~\ref{app:proof_rightwinner_1}.

\paragraph{Remark.} The astute readers may notice that the linearly \emph{negative} regret mainly comes from the bound of Eq.~\eqref{eqn:right_winner}. Since the runner upper is pulled much less often compared with the optimal advertiser (ad 1), the confidence bound of the runner up's pCTR is higher than the one of the optimal advertiser. Conditioning on the optimal advertiser wins the auction, this difference of confidence bound of the CTR estimates between the winner and runner up provides an additional lever to the price of the optimal advertiser. However, it is non-trivial to argue this claim is true and prove the $-\Omega(T)$ regret bound.

\section{Lower Bound Results}\label{sec:lower-bound}
In this section, we prove a $\Omega(\sqrt{T})$ regret lower bound for any stage-IC and IR mechanism. Indeed, the instance we construct in the lower bound proof still lies in the fixed valuation setting, however, the gap $\zeta$ (defined in Eq. ~\eqref{eq:min-gap}) is the order of $\frac{1}{\sqrt{T}}$. The proof of this lower bound follows the information-theoretical arguments and we carefully mitigates to the CTR prediction setting. 

\begin{theorem}
\label{thm:lb}
For any $T \geq 1$ and any stage-IC and IR auction $\cA$ (Definition~\ref{def:ir}), there exist an instance such that any online mechanism must incur $\Omega(\sqrt{T})$ regret.
\end{theorem}
\begin{proof}
Suppose there are four ads, each of which have value $1$.
We consider two different instances $\cI_1$ and $\cI_2$ and we let $\rho_{i, j}$ be the CTR of ad $j$ in instance $i$.
In $\cI_1$, we have $\rho_{1, 1} = \rho_{1, 2} = 1/2 + \eps / 2$ and $\rho_{1, 3} = \rho_{1, 4} = 1/2$.
In $\cI_2$, we have $\rho_{2, 1} = \rho_{2, 2} = 1/2 + \eps / 2$ and $\rho_{2, 3} = \rho_{2, 4} = 1/2$.
Let $\cA$ be any auction and let $R_i(T)$ denote the regret at time $T$ against instance $\cI_i$.
Let $r_{i,t}$ be the expected revenue received by $\cA$ at time $t$ and let $q_{i, t}$ be the probability that the ad chosen at time $t$ is in $\{1, 2\}$.
Note that $r_{i, t} \leq \rho_{i, 1} q_{i, t} + \rho_{i, 3} (1 - q_{i, t})$.
This is because the mechanism is IR so the expected revenue from showing ad $j$ can be no more than $\rho_{i,j}$.
Thus,
\[
    R_1(T) = \sum_{t=1}^T \frac{1+\eps}{2} - r_{1, t} \geq \sum_{t=1}^T \frac{\eps}{2} \cdot (1 - q_{1, t}).
\]
Analogously,
\[
    R_2(T) \geq \sum_{t=1}^T \frac{\eps}{2} \cdot q_{2,t}.
\]

Let $\cD_{i, t}$ be an independent realization of one draw from $\cI_i$ (i.e.~whether if an ad is shown, if it is clicked) and let $\cD_i = (\cD_{i,1}, \ldots, \cD_{i, T})$.
Let $\cD_{i, t}^\cA$ denote the distribution of the ads $\cA$ at time step $t$ when the instance is $\cI_i$
and $\cD_i^\cA = (\cD_{i,1}^\cA, \ldots, \cD_{i, T}^\cA)$.
Note that the arm chosen by $\cA$ at time $t$ is a randomized function of the realizations of $\cD_{i, 1}, \ldots, \cD_{i, t-1}$.
We have
\begin{align*}
    2\dtv{\cD_1^\cA}{\cD_2^\cA}^2
    & \leq 2\dtv{\cD_1}{\cD_2}^2 \\
    & \leq \KL{\cD_1}{\cD_2} \\
    & = T \cdot \KL{\cD_{1,1}}{\cD_{2,1}} \\
    & \leq 8T \eps^2.
\end{align*}
Here, the second inequality is Pinkser's Inequality and the last inequality is a straightforward calculation to verify that $\KL{\cD_{1,1}}{\cD_{2,1}} \leq 8 \eps^2$ for $\eps \in (0, 1/2]$.
Setting $\eps = \frac{1}{8\sqrt{T}}$, we have $\dtv{\cD_1^\cA}{\cD_2^\cA} \leq 1/2$.

Thus, we have that
\[
    R_1(T) + R_2(T) \ge \sum_{t=1}^T \frac{\eps}{2} \cdot (1-q_{1,t} + q_{2,t})
    \geq \frac{\eps T}{4} = \frac{\sqrt{T}}{32}.
\]
where in the inequality, we used that $\dtv{\cD_1^A}{\cD_2^A} \leq 1/2$ implies $q_{2,t} - q_{1,t} \geq -1/2$ for all $t \in [T]$.
Thus, $\max\{R_1(T), R_2(T)\} \geq \frac{\sqrt{T}}{64}$.
\end{proof}

\paragraph{Remark.} In fact, the lower bound proof doesn't utilize stage-IC property and we only need the advertisers don't overbid (bid cannot exceed value) and the mechanism is IR. Therefore, our lower bound result can be strengthen to a broader mechanism class.

\section{Non-Myopic Advertisers with Fixed Valuation}
\label{sec:global_ic_static}
In this section, we focus on the non-myopic setting.
Our goal is to design a global-IC mechanism to minimize regret. The adversarial valuation setting has been studied in prior work~\citep{DK09}, which shows $\Omega(T^{2/3})$ regret lower bound and provides an explore-then-commit algorithm to achieve matching regret upper bound.

Similar to the myopic case, we are still interested in the setting that the advertisers are non-myopic but the valuation are fixed, where there also exists a \emph{time-independent} constant gap $\zeta$ (defined in Eq.~\eqref{eq:min-gap}). 
We propose an online mechanism combining the
ideas of UCB and explore-then-commit algorithm, shown in Algorithm~\ref{alg:non_myopic_fixed}. For simplicity, we assume the seller can effectively elicit the true values $\{v_i \}$ through a global-IC mechanism beforehand. We will briefly dicuss how to make the algorithm work without this assumption in Remark~\ref{unknownvalue}.
The algorithm first runs pure exploration rounds until we find an arm whose lower confidence bound of the estimated eCPM is larger than the upper confidence bound of the estimated eCPMs of all the other arms. In the remaining rounds, we run naive VCG mechanism using the UCB estimates of CTRs observed from the initial exploration rounds, which is called ``exploitation" phase as we are not updating the estimated CTR anymore in these rounds.

\begin{algorithm}
\caption{Exlore-then-commit algorithm for fixed valuation setting}
\label{alg:non_myopic_fixed}
\begin{algorithmic}[1]
   \Repeat
    \State Show each ad $i \in [n]$ once (for free) and observe click; update $\hat{\rho}_{i}$ and $N_i$ accordingly. Let $\tilde{\rho}_i = \hat{\rho}_{i} + \sqrt{\frac{3 \log T}{2N_{i}}}$ (UCB), and $L_i = \hat{\rho}_{i} - \sqrt{\frac{3 \log T}{2N_{i}}}$ (LCB).
    \Until{finding a clear winner $i^*$ such that $v_{i^*} \cdot L_{i^*} > v_j \cdot \tilde{\rho}_{j}, \forall j \ne i^*$.}
\State Show each ad $i \in [n]$ once (for free) and observe click; initialize $\hat{\rho}_{i, 1} = \ind{\text{ad $i$ was clicked}}$ and $N_{i, 1} = 1$.
\For{each remaining round $t$}
\State Solicit bids $b_{i, t}$ for each advertiser $i \in [n]$.
\State Let $C_{i^*} = L_{i^*}$ and  $C_{j} = \tilde{\rho}_{j}$ for $j \ne i^*$.
\State Let $A_t \in \argmax_{i \in [n]} C_i \cdot b_{i, t}$ (winner) and $B_t \in \argsmax_{i \in [n]} C_i \cdot b_{i, t}$ (runner up). 
\State Show ad $A_t$. Let $X_t = \ind{\text{ad $A_t$ was clicked}}$ and charge $\frac{C_{B_t} \cdot b_{B_t, t}}{C_{A_t}} \cdot X_t$ to ad $A_t$ (other ads pay $0$). 
\EndFor
\end{algorithmic}
\end{algorithm}

First it is easy to see that the online mechanism shown in Algorithm~\ref{alg:non_myopic_fixed} is global-IC. The algorithm has two phases: exploration and exploitation. In the exploration phase, any bidding strategy will not affect the outcome nor the CTR learning procedure. In the exploitation phase, all auction parameters are fixed and thus each of the remaining rounds of this phase are independent auctions. Finally, we know that the mechanism is stage-IC in each individual round. Therefore we conclude that Algorithm~\ref{alg:non_myopic_fixed} is global-IC.

\begin{claim}
\label{claim:myerson2}
    The online mechanism proposed in Algorithm~\ref{alg:non_myopic_fixed} is global-IC.
\end{claim}

Our main result in this section is the following theorem which establishes a \emph{negative} regret for the mechanism proposed in Algorithm~\ref{alg:non_myopic_fixed}. Without loss of generality, we assume arm 1 is the optimal arm that has the highest eCPM score, i.e., $\rho_1 v_1 > \rho_j v_j, \forall j\neq 1$. Similar to the myopic setting, we define $\Delta_i = \frac{\rho_1 v_1 - \rho_i v_i}{v_i}$ for each advertiser $i\in [n]$ and assume the gap $\rho_1 v_1 - \max_{j\neq 1} \rho_j v_j$ (same as Eq.~\eqref{eq:min-gap}) is a \emph{time-independent} positive constant.
\begin{theorem}
\label{thm:static_non_myopic}
In the fixed valuation setting, the online mechanism shown in Algorithm~\ref{alg:non_myopic_fixed} obtains a regret of 
\[
R_T = -\Omega(T),
\]
as long as the gap $\rho_1 v_1 - \max_{j\neq 1} \rho_j v_j$ is a \emph{time-independent} positive constant.
\end{theorem}
Note that in Theorem~\ref{thm:static_non_myopic} and in the rest of this section, the notation $\Omega(\cdot)$ and $O(\cdot)$ also suppresses dependence on $n, \rho_i, v_i$.

We prove our \emph{negative} regret result in three steps.
First, the goal of exploration phase in Algorithm~\ref{alg:non_myopic_fixed} is to find a ``clear winner'', i.e.~a winner whose lower confidence bound is an upper bound on the upper confidence bound of all other ads.
In Lemma~\ref{lem:static_exploration}, we show that, with high probability, this takes at most $O(\log T)$ rounds and that the clear winner is the ad with the highest eCPM.
Next, we show, in Lemma~\ref{lem:static_positive_gap}, that there is a large gap between the UCB estimate of the CTR and the true CTR.
Finally, in Lemma~\ref{lem:static_negative_regret}, that the strictly positive gap from Lemma~\ref{lem:static_positive_gap} results in a negative regret in the exploitation phase.

\begin{lemma}
\label{lem:static_exploration}
Algorithm~\ref{alg:non_myopic_fixed} finds the clear winner $1$ in $O(\log T)$ rounds with probability $1 - O(\frac{1}{T})$. 
\end{lemma}
The proof of Lemma~\ref{lem:static_exploration} can be found in Appendix~\ref{app:global_ic_static}.

\begin{lemma}
\label{lem:static_positive_gap}
    After $O(\log T)$ rounds of observation, for any advertiser $i$, the difference between the upper confidence bound and the true CTR $\tilde{\rho}_i - \rho_i$ is at least $\Omega(1)$, with probability $1 - O(\frac{1}{T})$.
\end{lemma}

\begin{proof}
First, by the Hoeffding's inequality, $\hat{\rho}_{i} + \sqrt{\frac{3 \log T}{2N_{}}} / 2 > \rho_i$ with probability $1 - O(\frac{1}{T})$. Therefore, $\tilde{\rho}_i - \rho_i = \hat{\rho}_{i} + \sqrt{\frac{3 \log T}{2N_{}}}  - \rho_i > \sqrt{\frac{3 \log T}{2N_{}}} / 2 $ with probability $1 - O(\frac{1}{T})$. The gap $\sqrt{\frac{3 \log T}{2N_{}}} / 2 $ is lower bounded by a constant when $N = O(\log T)$.
\end{proof}

\begin{lemma}
\label{lem:static_negative_regret}
Algorithm~\ref{alg:non_myopic_fixed} achieves $-\Omega(T)$ regret in the exploitation phase after it finds the clear winner $1$.
\end{lemma}
\begin{proof}
First of all, by the Hoeffding's inequality, $\tilde{\rho}_i = \hat{\rho}_{i} + \sqrt{\frac{3 \log T}{2N_{}}} > \rho_i$, and $L_i = \hat{\rho}_{i} - \sqrt{\frac{3 \log T}{2N_{}}} < \rho_i$ for all advertiser $i$ with probability $1-O(\frac{1}{T})$. Now we only need to upper bound the regret in the case where $L_i < \rho_i < \tilde{\rho}_i $ for all $i$. Let the $\Omega(1)$ quantity guaranteed by Lemma~\ref{lem:static_positive_gap} be $c$, we have
\begin{align*}
R_T & = \E \left [ \sum_{t = 1}^T \left ( \rho_2 \cdot v_2 - \frac{\tilde{\rho}_{B_t} \cdot v_{B_t}}{L_1} \cdot X_t \right)  \right] \\
& =  \sum_{t = 1}^T  \left ( \rho_2 \cdot v_2 - \left( \frac{\tilde{\rho}_{B_t} \cdot v_{B_t}}{L_1} \right) \cdot \rho_1 \right ) \\
& =  \sum_{t = 1}^T  \left ( \rho_2 \cdot v_2 - \left( \frac{\tilde{\rho}_2 \cdot v_2}{L_1} \right) \cdot \rho_1 \right ) \\
& < \sum_{t = 1}^T  \left ( \rho_2 \cdot v_2 -  \frac{(\rho_2 + c) \cdot v_2}{L_1}  \cdot \rho_1 \right)  \\
& < -c \cdot v_2 \cdot T \\
& = -\Omega(T). \qedhere
\end{align*}
\end{proof}

Putting Lemma~\ref{lem:static_exploration}, ~\ref{lem:static_positive_gap}, and ~\ref{lem:static_negative_regret} together, we have the final statement that Algorithm~\ref{alg:non_myopic_fixed} achieves $-\Omega(T)$ regret when the valuations are static and the gap between optimal ad and suboptimal ads is a time-independent positive constant.
This completes the proof of Theorem~\ref{thm:static_non_myopic}.

Note that the algorithm in this section does not subsume Algorithm~\ref{alg:ucb_auction} as this algorithm only works when the auctioneer knows all advertiers' values are static in advance. 

\begin{remark}
    \label{unknownvalue}
    It is worth noting that we make an assumption that the autioneer can effectively elicit true values of all ads beforehand. Given that Algorithm~\ref{alg:non_myopic_fixed} utilizes these true values only in Step 3 (namely, the exploration termination condition), and considering our assumption of a positive constant gap $\zeta$, we can implement Algorithm~\ref{alg:non_myopic_fixed} without prior knowledge of true values by simply fixing the number of exploration rounds as a large enough constant depending on $\zeta$. Consequently, it becomes apparent that this approach enables the algorithm to achieve both global-IC and a regret of $-\Omega(T)$ in terms of revenue.
\end{remark}
\section{Conclusions and Future Work}\label{sec:conclusion}
In this paper, we designed online learning algorithms for pay-per-click auctions. When the advertisers are myopic, we designed an online mechanism based on UCB that has $O(\sqrt{T})$ regret in the worst case and $-\Omega(T)$ regret when the values are static.
In the setting where the advertisers are not myopic, we designed an online auction based on explore-then-commit and UCB that also achieves $-\Omega(T)$ regret.

We conclude this paper with two possible avenues for further research.
First, we raise the question of designing online mechanisms for advertisers that are neither fully myopic nor fully non-myopic.
One way to formalize this is to assume that advertisers wish to maximize their $\gamma$-discounted long-term utility where $\gamma = 0$ corresponds to the myopic setting and $\gamma = 1$ corresponds to the fully non-myopic setting.
If the values are chosen adversarially then the optimal regret when $\gamma = 0$ is $\wtilde{\Theta}(\sqrt{T})$ and the optimal regret when $\gamma = 1$ is $\Theta(T^{2/3})$ \cite{DK09}.
We leave it as an open question to design a mechanism with $o(T^{2/3})$ regret when $\gamma \in (0, 1)$. A second question is to consider a contextual version of the problem where the CTR may depend on some context.

\section{Acknowledgments}
Part of work was done when Zixin Zhou was a Student Researcher at Google Research, Mountain View. We also thank Aranyak Mehta for his insightful feedback.
\bibliographystyle{icml2023}

\newpage
\appendix
\onecolumn

\section{Standard Facts}
\label{app:facts}

\begin{lemma}[Hoeffding's Inequality]
\label{lem:hoeffding}
Let $X_1, \ldots, X_k$ be independent random variables such that $X_i \in [0, 1]$ for all $i \in [k]$.
Let $S_k = \frac{1}{k}\sum_{i=1}^k X_i$.
Then
\begin{align*}
    \Pr[S_k - \E[S_k] > t] & \leq \exp\left( - 2kt^2 \right) \\
    \Pr[S_k - \E[S_k] < -t] & \leq \exp\left( - 2kt^2 \right).
\end{align*}
\end{lemma}

\begin{fact}
    \label{fact:sqrt_sum_bound}
    For all $K \geq 1$, $\sum_{k=1}^K \frac{1}{\sqrt{k}} \leq 2 \sqrt{K}$.
\end{fact}
\begin{proof}
    We have that
    \begin{align*}
        \sum_{k=1}^K \frac{1}{\sqrt{k}}
        \leq \int_1^{K+1} \frac{1}{\sqrt{x}} \, \dd x
        = 2\sqrt{K+1} - 2
        \leq 2\sqrt{K},
    \end{align*}
    where the last inequality is Fact~\ref{fact:sqrt_bound}.
\end{proof}

\begin{fact}
    \label{fact:sqrt_bound}
    For all $x \geq 0$, $\sqrt{x+1} \leq \sqrt{x} + 1$.
\end{fact}
\begin{proof}
    The inequality in the claim is equivalent to $x+1 \leq x + 2\sqrt{x} + 1$ which is true for all $x \geq 0$.
\end{proof}
\section{Missing Proofs from Section~\ref{sec:ucb_no_regret}}
\label{app:ucb_no_regret}

\subsection{Proof of Proposition~\ref{prop:myerson1}}\label{app:ucb-ic}
\ucbisic*
\begin{proof}
Monotonicity is clear since line~\ref{line:ucb_ranking} ranks by the score $\tilde{\rho}_{i, t} \cdot b_{i, t}$.
The score is non-decreasing in $b_{i, t}$ since $\tilde{\rho}_{i,t}$, as computed in line~\ref{line:ucb_score}, is strictly positive.

Let $p_{i, t}(b_t)$ be advertiser $i$'s expected payment when the bids are $b_t$ and let $x_{i, t}(b_t)$ be the probability that ad $i$ is clicked when the bids are $b_t$.
We need to show that
\begin{equation}
\label{eqn:myerson}
p_{i, t}(b_t) = b_{i, t} \cdot x_{i, t}(b_t) - \int_0^{b_{i, t}} x_{i, t}(z, b_{-i, t}) \, \dd z.
\end{equation}

First, suppose that $i \notin \argmax_{i' \in [n]} \tilde{\rho}_{i, t} \cdot b_{i, t}$.
Then line~\ref{line:ucb_payment} states that $p_{i, t}(b_t) = 0$.
Moreoever, the RHS of Eq.~\ref{eqn:myerson} is also $0$ since $x_{i, t}(z, b_{-i, t}) = 0$ if $\tilde{\rho}_{i, t} \cdot z < \max_{i' \neq i} \tilde{\rho}_{i', t} \cdot b_{i', t}$.

On the other hand, suppose that $i \in \argmax_{i' \in [n]} \tilde{\rho}_{i, t} \cdot b_{i, t}$.
Let $A_t, B_t$ be as in line~\ref{line:ucb_ranking}.
Let $Y_{i, t} = \ind{A_t = i}$.
Note that $\E[Y_{i, t}] = \Pr[Y_{i, t} = 1] = 1$ if $\tilde{\rho}_{i,t} \cdot b_{i, t} > \tilde{\rho}_{B_t, t} \cdot b_{B_t, t}$.
It is straightforward to check that 
\begin{equation}
\label{eqn:payment}
p_{i,t}(b_t) = \frac{\tilde{\rho}_{B_t, t} \cdot b_{B_t, t}}{\tilde{\rho}_{i, t}} \cdot \E[X_t \cdot Y_{i, t}]
\end{equation}
Next, observe that $x_{i, t}(z, b_{-i, t}) = \rho_t = \E[X_t \mid Y_{i, t} = 1]$ if $\tilde{\rho}_{i, t} \cdot z > \max_{i' \neq i} \tilde{\rho}_{i', t} \cdot b_{i', t}$.
Thus, noting that $x_{i, t}(b_t) = \E[X_t Y_{i, t}]$, we have
\begin{equation}
\label{eqn:myerson_1}
\begin{aligned}
b_{i, t} & \cdot x_{i,t}(b_t) - \int_0^{b_{i, t}} x_{i, t}(z, b_{-i, t}) \, \dd z \\
& = b_{i, t} \cdot \E[X_t \cdot Y_{i, t}] - \\ 
& \quad \quad  \E[X_t \mid Y_{i, t} = 1] \left(b_{i, t} - \frac{\tilde{\rho}_{B_t, t} \cdot b_{B_t, t}}{\tilde{\rho}_{i, t}}\right).
\end{aligned}
\end{equation}
If $\tilde{\rho}_{i,t} \cdot b_{i, t} = \tilde{\rho}_{B_t, t} \cdot b_{B_t, t}$ then Eq.~\eqref{eqn:myerson_1} is exactly Eq.~\eqref{eqn:payment} so Eq.~\eqref{eqn:myerson} is satisfied.
If $\tilde{\rho}_{i,t} \cdot b_{i, t} > \tilde{\rho}_{B_t, t} \cdot b_{B_t, t}$ then $\E[X_t \mid Y_{i, t} = 1] = \E[X_t \mid Y_{i, t} = 1] \cdot \Pr[Y_{i, t} = 1] = \E[X_t \cdot Y_{i, t}]$ where the first equality is because $\Pr[Y_{i, t} = 1] = 1$.
So Eq.~\eqref{eqn:myerson} is also satisfied in this case.
\end{proof}

\subsection{Proof of Theorem~\ref{thm:sqrt_ub}}\label{app:sqrt_ucb}
\sqrtub*
\begin{proof}
    Note that we incur regret $M \cdot n$ to initialize each of the UCB estimates.

    Let $\cE = \left\{\forall i \in [n], \forall t \in [T], \tilde{\rho}_{i,t} - \rho_i \in \left[0, 2 \sqrt{\frac{3 \log(2nT)}{2N_{i, t}}}\right] \right\}$.
    By Lemma~\ref{lem:ucb_estimator}, we have $\Pr[\cE] \geq 1 - \frac{1}{n^2T^2}$.
    Let $s_t \in \argsmax_{i \in [n]} \rho_{i,t} v_{i, t}$.
    Let $T_i = \sum_{t=1}^T \ind{A_t = i}$.
    On the event $\cE$, we have
    \begin{align*}
        \sum_{t=1}^T \sum_{i=1}^n \sum_{j=1}^n & \left( \rho_{s_t, t} v_{s_t, t} - \tilde{\rho}_{j, t} v_{j, t} \cdot \frac{\rho_i}{\tilde{\rho}_{i, t}} \right) \cdot \ind{A_t = i, B_t = j} \\
        & \leq 
        \sum_{t=1}^T \sum_{i=1}^n \sum_{j=1}^n \left( \rho_{s_t, t} v_{s_t, t} - \rho_{s_t,t} v_{s_t, t} \cdot \frac{\rho_i}{\rho_i + 2\sqrt{\frac{3 \log(2nT)}{2N_{i,t}}}} \right) \cdot \ind{A_t = i, B_t = j} \\
        & \leq M \sum_{t=1}^T \sum_{i=1}^n \sum_{j=1}^n \left( 1 - \frac{\rho_i}{\rho_i + 2\sqrt{\frac{3 \log(2nT)}{2N_{i,t}}}} \right) \cdot \ind{A_t = i, B_t = j} \\
        & = M \sum_{t=1}^T \sum_{i=1}^n \left( 1 - \frac{\rho_i}{\rho_i + 2\sqrt{\frac{3 \log(2nT)}{2N_{i,t}}}} \right) \cdot \ind{A_t = i} \\
        & = M \sum_{i=1}^n \sum_{t=1}^{T_i} \left( 1 - \frac{\rho_i}{\rho_i + 2\sqrt{\frac{3 \log(2nT)}{2t}}} \right) \\
        & \leq M  \sum_{i=1}^n \sum_{t=1}^T \left( 1 - \frac{\rho_i}{\rho_i + 2\sqrt{\frac{3 \log(2nT)}{2t}}} \right) \\
        & \leq M \sum_{i=1}^n \sum_{t=1}^T \frac{2}{\rho_i} \sqrt{\frac{3 \log(2nT)}{2t}} \\
        & \leq M \sum_{i=1}^n \frac{\sqrt{24T \log(2nT)}}{\rho_i}.
    \end{align*}
    In the first inequality we used that, on the event $\cE$, $\tilde{\rho}_{j, t} v_{j, t} = \smax_{k} \tilde{\rho}_{k,t} v_k \geq \smax_{k} \rho_{k, t} v_{k,t} = \rho_{s_t, t} v_{s_t, t}$,
    in the fourth inequality, we used that $1 - \frac{x}{x+y} \leq \frac{y}{x}$ for $x, y > 0$,
    and in the last inequality, we used Fact~\ref{fact:sqrt_sum_bound}.
    Finally, on the event $\cE^c$, we can use a trivial bound of $M$ on the regret for each time step.
    We conclude that the regret is at most $M \cdot n + M \cdot \Pr[\cE^c] + M \cdot \sum_{i=1}^n \frac{\sqrt{24T \log(2nT)}}{\rho_i} \Pr[\cE] \leq \frac{M}{T} + M \cdot \sum_{i=1}^n \frac{\sqrt{24 T \log(2nT)}}{\rho_i}$.
\end{proof}

In this section, we make use of a couple standard lemmas that assert that $\tilde{\rho}_{i,t}$ is a good upper bound on the mean.
\begin{lemma}
\label{lem:rho_bound}
Let $a > 1$.
For $i \in [n]$, with probability $1 - \frac{1}{n^{a}T^{a-1}}$, $|\hat{\rho}_{i, t} - \rho_i| \leq \sqrt{\frac{a \log(2nT)}{2N_{i,t}}}$ for all $t \in [T]$.
\end{lemma}
\begin{proof}
    We use a standard coupling argument.
    Let $\wtilde{X}_{i,1}, \ldots, \wtilde{X}_{i, T}$ be independent $\Ber(\rho_i)$ random variables.
    Let $\wtilde{\rho}_{i, k} = \frac{1}{k} \sum_{\ell = 1}^k \wtilde{X}_{i, \ell}$.
    We then couple Algorithm~\ref{alg:ucb_auction} by setting $X_t = \wtilde{X}_{A_t, N_{i,t}+1}$.
    Then, by Hoeffding's Inequality (Lemma~\ref{lem:hoeffding}), we have $\Pr\left[ |\wtilde{\rho}_{i,k} - \rho_i| > r\right] \leq 2\exp\left( - 2kr^2 \right)$.
    To make the RHS less than $\frac{1}{n^a T^a}$, we take $r = \sqrt{\frac{a \log(2nT)}{2k}}$.
    Taking a union bound over all $k \in [T]$ proves the claim.
\end{proof}
\begin{lemma}
    \label{lem:ucb_estimator}
    With probability $1 - \frac{1}{n^2T^2}$,
    $\tilde{\rho}_{i,t} - \rho_i \in \left[0, 2\sqrt{\frac{3 \log(2nT)}{2N_{i, t}}}\right]$ for all $i \in [n]$ and $t \in [T]$.
\end{lemma}
\begin{proof}
    Recall that $\tilde{\rho}_{i, t} = \hat{\rho}_{i, t} + \sqrt{\frac{3 \log(2nT)}{2N_{i,t}}}$ and apply Lemma~\ref{lem:rho_bound} with $a = 3$ with a union bound over $i \in [n]$.
\end{proof}
Recall that for $i \in \{2, \ldots, n\}$, $\Delta_i = \frac{\rho_1 v_1 - \rho_i v_i}{v_i}$.
\begin{lemma}
    \label{lem:bound_Nit}
    With probability $1 - \frac{1}{n^2T^2}$, for all $i \in \{2, \ldots, n\}$, $N_{i,T} \leq \frac{6 \log(2nT)}{\Delta_i^2}$.
\end{lemma}
\begin{proof}
    Suppose that $\tilde{\rho}_{i,t} - \rho_i \in \left[0, 2\sqrt{\frac{3 \log(2nT)}{2N_{i, t}}}\right]$ for all $i \in [n]$ and $t \in [T]$ (which happens with probability $1 - \frac{1}{n^2T^2}$ by Lemma~\ref{lem:ucb_estimator}).
    If $i \in \{2, \ldots, n\}$ and $N_{i, t} \geq \frac{6 \log(2nT)}{\Delta_i^2}$ then $\tilde{\rho}_{i,t} v_i \leq \rho_i v_i + 2\sqrt{\frac{3 \log(2nT)}{2N_{i,t}}} v_i < \rho_i v_i + \Delta_i v_i = \rho_1 v_1 < \tilde{\rho}_{1,t}$.
    So arm $i$ is not chosen and thus, $N_{i,t} \leq \frac{6 \log(2nT)}{\Delta_i^2}$ for all $t$ (and, in particular, for $t = T$).
\end{proof}

\subsection{Proof of Lemma~\ref{lem:wrong_winner}}
\label{app:wrong_winner_proof}
\begin{lemma}
$\E \left [\sum_{t=1}^T \sum_{i=2}^n \sum_{j=1}^n \left ( \rho_s v_s - \tilde{\rho}_{j,t} v_j \cdot \frac{X_t}{\tilde{\rho}_{i, t}} \right) \mathbb{I} \left \{ A_t = i, B_t = j\right \}  \right]  \leq \frac{2}{T} + \sum_{i=2}^n \frac{12 \rho_s v_s}{\rho_1} \frac{\log(2nT)}{\Delta_i} + \frac{\rho_s v_s}{\rho_i} \frac{\sqrt{6 \log(2nT)}}{n^2T}$.
\end{lemma}
\begin{proof}
Let $\cE_{i,t} = \left\{\tilde{\rho}_{1,t} \geq \rho_1, \tilde{\rho}_{i,t} \leq \rho_i + 2\sqrt{\frac{3 \log(2nT)}{2N_{i,t}}}\right\}$.
First, we write
\begin{align}
\sum_{t = 1}^T \sum_{i = 2}^n \sum_{j=1}^n & \left ( \rho_s v_s - \tilde{\rho}_{j,t} v_j \cdot \frac{X_t}{\tilde{\rho}_{i,t}} \right) \mathbb{I} \left \{ A_t = i, B_t = j\right \} \notag \\
& \le \sum_{t = 1}^T \sum_{i = 2}^n \sum_{j=1}^n  \rho_s v_s   \mathbb{I} \left \{ \tilde{\rho}_{1,t} < \rho_1 \right \} \label{eqn:wrongwinner_1}  \\
& + \sum_{t = 1}^T \sum_{i = 2}^n \sum_{j=1}^n  \rho_s v_s   \mathbb{I} \left \{ \tilde{\rho}_{i,t} > \rho_i+ 2 \sqrt{\frac{3 \log(2nT)}{2N_{i,t}}} \right \} \label{eqn:wrongwinner_2} \\
& + \sum_{t = 1}^T \sum_{i = 2}^n \sum_{j=1}^n \left ( \rho_s v_s - \tilde{\rho}_{j,t} v_j \cdot \frac{X_t}{\tilde{\rho}_{i,t}} \right) \mathbb{I} \left \{ A_t = i, B_t = j, \cE_{i,t} \right \}. \label{eqn:wrongwinner_3}
\end{align}
We now bound Eq.~\eqref{eqn:wrongwinner_1}, Eq.~\eqref{eqn:wrongwinner_2}, and Eq.~\eqref{eqn:wrongwinner_3} separately.
\begin{claim}
\label{claim:wrongwinner_1}
$\sum_{t = 1}^T \sum_{i = 2}^n \sum_{j=1}^n  \rho_s v_s \Pr\left[ \tilde{\rho}_{1,t} < \rho_1 \right] \leq \frac{1}{T}$.
\end{claim}
\begin{proof}
    By Lemma~\ref{lem:ucb_estimator}, we have $\Pr[\tilde{\rho}_{1, t} < \rho_1] \leq \frac{1}{n^2T^2}$.
    So, $\sum_{t = 1}^T \sum_{i = 2}^n \sum_{j=1}^n  \rho_s v_s \Pr\left[ \tilde{\rho}_{1,t} < \rho_1 \right] \leq \frac{n^2T}{n^2T^2} =  \frac{1}{T}$.
\end{proof}
\begin{claim}
\label{claim:wrongwinner_2}
$\sum_{t = 1}^T \sum_{i = 2}^n \sum_{j=1}^n  \rho_s v_s \Pr\left[ \tilde{\rho}_{i,t} > \rho_i+ 2 \sqrt{\frac{3 \log(2nT)}{2N_{i,t}}} \right] \leq \frac{1}{T}$.
\end{claim}
\begin{proof}
    Similar to Claim~\ref{claim:wrongwinner_2}, this follows from Lemma~\ref{lem:ucb_estimator} which gives $\Pr\left[ \tilde{\rho}_{i,t} > \rho_i+ 2 \sqrt{\frac{3 \log(2nT)}{2N_{i,t}}} \right] \leq \frac{1}{n^2T^2}$.
\end{proof}
The bound for Eq.~\eqref{eqn:wrongwinner_3} requires a bit more work and we relegate the proof of the next claim to Appendix~\ref{subsec:wrongwinner_3}.
\begin{claim}
\label{claim:wrongwinner_3}
We have that
\[
\E\left[ \sum_{t = 1}^T \sum_{i = 2}^n \sum_{j=1}^n \left ( \rho_s v_s - \tilde{\rho}_{j,t} v_j \cdot \frac{X_t}{\tilde{\rho}_{i,t}} \right) \mathbb{I} \left \{ A_t = i, B_t = j, \cE_{i, t} \right \} \right] \leq \sum_{i=2}^n \frac{12 \rho_s v_s}{\rho_1} \frac{\log(2nT)}{\Delta_i} + \frac{\rho_s v_s}{\rho_i} \frac{\sqrt{6 \log(2nT)}}{n^2T}.
\]
\end{claim}
The lemma follows by combining the previous three claims.
\end{proof}

\subsubsection{Proof of Claim~\ref{claim:wrongwinner_3}}
\begin{proof}
Recall that $\cE_{i,t} = \left\{\tilde{\rho}_{1,t} \geq \rho_1, \tilde{\rho}_{i,t} \leq \rho_i + 2\sqrt{\frac{3 \log(2nT)}{2N_{i,t}}}\right\}$.
We have
\label{subsec:wrongwinner_3}
\begin{align}
    \E & \left [\sum_{t = 1}^T \sum_{i = 2}^n \sum_{j=1}^n  \left ( \rho_s v_s - \tilde{\rho}_{j,t} v_j \cdot \frac{X_t}{\tilde{\rho}_{i,t}} \right) \mathbb{I} \left \{ A_t = i, B_t = j, \cE_{i,t} \right \} \right ] \notag \\
    & \le \E \left [\sum_{t = 1}^T \sum_{i = 2}^n \sum_{j=1}^n  \left ( \rho_s v_s - \rho_s v_s \cdot \frac{\rho_i}{\rho_i + 2 \sqrt{\frac{3 \log(2nT)}{2N_{i,t}}} } \right) \mathbb{I} \left \{ A_t = i, B_t = j, \cE_{i, t} \right \} \right ] \notag \\
    & \le \E \left [\sum_{t = 1}^T \sum_{i = 2}^n \sum_{j=1}^n  \left ( \rho_s v_s - \rho_s v_s \cdot \frac{\rho_i}{\rho_i + 2 \sqrt{\frac{3 \log(2nT)}{2N_{i,t}}} } \right) \mathbb{I} \left \{ A_t = i, B_t = j\right \} \right ] \notag \\
    & \le \E \left [\sum_{t = 1}^T \sum_{i = 2}^n \left ( \rho_s v_s - \rho_s v_s \cdot \frac{\rho_i}{\rho_i + 2 \sqrt{\frac{3 \log(2nT)}{2N_{i,t}}} } \right) \mathbb{I} \left \{ A_t = i\right \} \right ] \notag \\
    & = \E \left [\sum_{t = 1}^T \sum_{i = 2}^n \left ( \rho_s v_s - \rho_s v_s \cdot \frac{\rho_i}{\rho_i + 2 \sqrt{\frac{3 \log(2nT)}{2N_{i,t}}} } \right) \mathbb{I} \left \{ A_t = i\right \} \right ] \notag \\
    & = \E \left [\sum_{t = 1}^T \sum_{i = 2}^n \left ( \rho_s v_s\left( 1 - \frac{\rho_i}{\rho_i + 2 \sqrt{\frac{3 \log(2nT)}{2N_{i,t}}} } \right) \right) \mathbb{I} \left \{ A_t = i\right \} \right ] \notag \\
    & = \E \left [\sum_{t = 1}^T \sum_{i = 2}^n \left ( \rho_s v_s\left( 1 - \frac{1}{1 + \frac{2}{\rho_i} \sqrt{\frac{3 \log(2nT)}{2N_{i,t}}} } \right) \right) \mathbb{I} \left \{ A_t = i\right \} \right ] \notag \\
    & \leq \E \left [\sum_{t = 1}^T \sum_{i = 2}^n \frac{2 \rho_s v_s}{\rho_i} \sqrt{\frac{3 \log (2nT)}{2N_{i, t}}} \mathbb{I} \left \{ A_t = i\right \} \right] \notag \\
    & = \E \left [\sum_{t = 1}^T \sum_{i = 2}^n \frac{2 \rho_s v_s}{\rho_i} \sqrt{\frac{3 \log(2nT)}{2N_{i, t}}} \mathbb{I} \left \{ A_t = i, N_{i, T} \leq \frac{6 \log(2nT)}{\Delta_i^2} \right \} \right ] \label{eqn:ucb_4} \\
    & \quad \quad + \E \left [\sum_{t = 1}^T \sum_{i = 2}^n \frac{2 \rho_s v_s}{\rho_i} \sqrt{\frac{3 \log(2nT)}{2N_{i, t}}} \mathbb{I} \left \{ A_t = i, N_{i, T} > \frac{6 \log(2nT)}{\Delta_i^2} \right \} \right ] \label{eqn:ucb_5}
\end{align}
In the first inequality above, we used that (i) $\E[X_t | \ind{A_t = i} ] = \rho_i$ (and that the conditional expectation is independent of everything else)
and that (ii) $\tilde{\rho}_{j, t} v_j \geq \tilde{\rho}_{1, t} v_1 \geq \rho_1 v_1 \geq \rho_s v_s$.
Here, the first inequality is because $A_t \neq 1$ so the runner-up score is atleast the runner-up score of ad $1$,
the second inequality is because we conditioned on the event $\tilde{\rho}_{1,t} \geq \rho_1$,
and the third inequality is because $1 \in \argmax_{i \in [n]} \rho_i v_i$.
To bound Eq.~\eqref{eqn:ucb_4}, note that
\begin{align*}
\sum_{t = 1}^T \sum_{i = 2}^n \frac{2 \rho_s v_s}{\rho_i} \sqrt{\frac{3 \log(2nT)}{2N_{i, t}}} & \mathbb{I} \left \{ A_t = i, N_{i, T} \leq \frac{6 \log(2nT)}{\Delta_i^2} \right \} \\
& \leq
\sum_{i = 2}^n \sum_{k=1}^{\frac{6 \log (2nT)}{\Delta_i^2}}
\frac{2 \rho_s v_s}{\rho_i} \sqrt{\frac{3 \log(2nT)}{2k}} \mathbb{I} \left \{ A_t = i, N_{i, T} \leq \frac{6 \log(2nT)}{\Delta_i^2} \right \} \\
& \leq \frac{12 \rho_s v_s}{\rho_1} \frac{\log(2nT)}{\Delta_i},
\end{align*}
where the last inequality is by Fact~\ref{fact:sqrt_sum_bound}.

Finally, to bound Eq.~\eqref{eqn:ucb_5}, we use the trivial bound $1/\sqrt{N_{i, t}} \leq 1$ and Lemma~\ref{lem:bound_Nit} to get that
\begin{align*}
\text{Eq.~\eqref{eqn:ucb_5}}
& \leq \sum_{t=1}^T \sum_{i=1}^n \frac{2 \rho_s v_s}{\rho_i} \sqrt{\frac{3\log(2nT)}{2}} \Pr\left[N_{i, T} > \frac{6 \log(2nT)}{\Delta_i^2} \right] \\
& \leq \sum_{t=1}^T \sum_{i=1}^n \frac{2 \rho_s v_s}{\rho_i} \sqrt{\frac{3\log(2nT)}{2}} \frac{1}{n^2T^2} \\
& \leq \sum_{i=1}^n \frac{2 \rho_s v_s}{\rho_i} \sqrt{\frac{3\log(2nT)}{2}} \frac{1}{n^2T}.
\end{align*}
The proof is complete.
\end{proof}

\subsection{Proof of Lemma~\ref{lem:right_winner}}
\label{app:right_winner_proof}
\rightwinner*
\begin{proof}
We have that
\begin{align*}
	\E & \left [\sum_{t = 1}^T \sum_{j=1}^n \left ( \rho_s v_s - \tilde{\rho}_{j,t} v_j \cdot \frac{X_t}{\tilde{\rho}_{1,t}} \right ) \mathbb{I} \left \{ A_t = 1, B_t = j\right \}  \right] \\
	& \le  \E \left [\sum_{t = 1}^T \sum_{j=1}^n \left ( \rho_s v_s - \tilde{\rho}_{j,t} v_j \cdot \frac{X_t}{\tilde{\rho}_{1,t}} \right ) \mathbb{I} \left \{ A_t = 1, B_t = j, \tilde{\rho}_{s, t} \ge \rho_s + 0.08 \Delta_s, \tilde{\rho}_{1,t} \le \rho_1 + 2 \sqrt{\frac{3 \log(2nT)} {2N_{1, t}}}\right \}  \right]  + \\
	& \quad \quad \quad \E \left [\sum_{t = 1}^T \sum_{j=1}^n \left ( \rho_s v_s  \right ) \mathbb{I} \left \{ \tilde{\rho}_{s, t} < \rho_s + 0.08 \Delta_s\right \}  \right] + \\
	& \quad \quad \quad  \quad \quad  \E \left [\sum_{t = 1}^T \sum_{j=1}^n \left ( \rho_s v_s \right ) \mathbb{I} \left \{\tilde{\rho}_{1,t} > \rho_1 + 2 \sqrt{\frac{3 \log(2nT)} {2N_{1, t}}}\right \}  \right].
\end{align*}

We begin with the second and third term first since the proofs are short.
\begin{claim}
\label{claim:Delta_gap}
For any $t \in [T]$, $\Pr[\tilde{\rho}_{s,t} \geq \rho_s + 0.08 \Delta_s] \geq 1 - \frac{2}{nT}$.
\end{claim}
\begin{proof}
    Using Lemma~\ref{lem:rho_bound} with $a = 2$ gives that, with probability $1 - \frac{1}{n^2T}$, $\hat{\rho}_{s, t} \geq - \sqrt{\frac{\log(2nT)}{N_{s,t}}}$ for all $t \in [T]$.
    On this event, we have $\tilde{\rho}_{s,t} \geq \sqrt{\frac{\log(2nT)}{N_{i,t}}} \cdot \left(\sqrt{1.5} - 1 \right) > 0.2 \sqrt{\frac{\log(2nT)}{N_{i,t}}}$.
    Next, using Lemma~\ref{lem:bound_Nit}, we have $N_{s, t} \leq \frac{6 \log(2nT)}{\Delta_i^2}$ for all $t \in [T]$ with probability $1 - \frac{1}{n^2T^2}$.
    Condition on the above two events, we have $\tilde{\rho}_{s,t} \geq \rho_s + 0.08 \Delta_s$ with probability at least $1 - \frac{2}{nT}$.
\end{proof}
\begin{claim}
\label{claim:rightwinner_2}
$\E \left [\sum_{t = 1}^T \sum_{j=1}^n \left ( \rho_s v_s  \right ) \mathbb{I} \left \{ \tilde{\rho}_{s, t} < \rho_s + 0.08 \Delta_s\right \}  \right] \leq 2 \rho_s v_s$.
\end{claim}
\begin{proof}
    Using Claim~\ref{claim:Delta_gap}, we conclude that
    \[
        \E \left [\sum_{t = 1}^T \sum_{j=1}^n \left ( \rho_s v_s  \right ) \mathbb{I} \left \{ \tilde{\rho}_{s, t} < \rho_s + 0.08 \Delta_s\right \}  \right] \leq 2 \rho_s v_s,
    \]
    as desired.
\end{proof}
\begin{claim}
\label{claim:rightwinner_3}
$\E \left [\sum_{t = 1}^T \sum_{j=1}^n \left ( \rho_s v_s \right ) \mathbb{I} \left \{\tilde{\rho}_{1,t} > \rho_1 + 2 \sqrt{\frac{3 \log(2nT)} {2N_{1, t}}}\right \}  \right] \leq \rho_s v_s$.
\end{claim}
\begin{proof}
    This follows easily from Lemma~\ref{lem:ucb_estimator} which implies that $\Pr\left[\tilde{\rho}_{1,t} > \rho_1 + 2 \sqrt{\frac{3 \log(2nT)} {2N_{1, t}}}\right] \leq \frac{1}{n^2T^2}$.
\end{proof}
We now bound the first term. Let $\cE_t = \left\{\tilde{\rho}_{s,t} \geq \rho_s + 0.08 \Delta_s, \tilde{\rho}_{1,t} \leq \rho_1 + 2\sqrt{\frac{3 \log(2nT)}{2N_{1,t}}} \right\}$.
\begin{claim}
    \label{claim:rightwinner_1}
    \begin{align*}
        \E & \left [\sum_{t = 1}^T \sum_{j=1}^n \left ( \rho_s v_s - \tilde{\rho}_{j,t} v_j \cdot \frac{X_t}{\tilde{\rho}_{1,t}} \right ) \mathbb{I} \left \{ A_t = 1, B_t = j, \cE_t \right \}  \right] \\
        & \leq \frac{9000 \rho_s^3 v_s \log T}{\rho_1^2 \Delta_s^2} - 0.05 \Delta_s v_s \left( T - \left( \frac{5}{nT} + \frac{9000 \rho_s^2 \log(2nT)}{\rho_1^s \Delta_s^2} + \sum_{i=2}^n \frac{12 \log(2nT)}{\Delta_i^2} \right) \right)
    \end{align*}
\end{claim}
The proof of Claim~\ref{claim:rightwinner_1} can be found in Appendix~\ref{app:proof_rightwinner_1}.
Combining Claim~\ref{claim:rightwinner_2}, Claim~\ref{claim:rightwinner_3}, Claim~\ref{claim:rightwinner_1} completes the proof of the lemma.
\end{proof}

\subsubsection{Proof of Claim~\ref{claim:rightwinner_1}}
\label{app:proof_rightwinner_1}
\begin{proof}
Recall that $\cE_t = \left\{\tilde{\rho}_{s,t} \geq \rho_s + 0.08 \Delta_s, \tilde{\rho}_{1,t} \leq \rho_1 + 2\sqrt{\frac{3 \log(2nT)}{2N_{1,t}}} \right\}$.
Let $T_0 = \frac{9000 \rho_s^2 \log(2nT)}{\rho_1^2 \Delta_s^2}$.
\begin{align*}
	\E & \left [\sum_{t = 1}^T \sum_{j=1}^n \left ( \rho_s v_s - \tilde{\rho}_{j,t} v_j \cdot \frac{X_t}{\tilde{\rho}_{1,t}} \right ) \mathbb{I} \left \{ A_t = 1, B_t = j, \cE_t \right \}  \right]  \\
	 &  \le \E  \left [\sum_{t = 1}^T \sum_{j=1}^n \left ( \rho_s v_s - (\rho_s + 0.08 \Delta_s) v_s \cdot \frac{X_t}{\rho_1 + 2 \sqrt{\frac{3 \log(2nT)} {2N_{1, t}}}} \right ) \mathbb{I} \left \{ A_t = 1, B_t = j, \cE_t \right \}  \right]  \\
  	 &  = \E  \left [\sum_{t = 1}^T \sum_{j=1}^n \left ( \rho_s v_s - (\rho_s + 0.08 \Delta_s) v_s \cdot \frac{\rho_1}{\rho_1 + 2 \sqrt{\frac{3 \log(2nT)} {2N_{1, t}}}} \right ) \mathbb{I} \left \{ A_t = 1, B_t = j, \cE_t \right \}  \right]  \\
    &  = \E  \left [\sum_{t = 1}^T  \left ( \rho_s v_s - (\rho_s + 0.08 \Delta_s) v_s \cdot \frac{\rho_1}{\rho_1 + 2 \sqrt{\frac{3 \log(2nT)} {2N_{1, t}}}} \right ) \mathbb{I} \left \{ A_t = 1, \cE_t \right \}  \right]  \\
    &  \le \rho_s v_s \cdot T_0 + \E  \left [\sum_{t = 1}^T  \left ( \rho_s v_s - (\rho_s + 0.08 \Delta_s) v_s \cdot \frac{\rho_1}{\rho_1 + 2 \sqrt{\frac{3 \log(2nT)} {2T_0}}} \right ) \mathbb{I} \left \{ A_t = 1, N_{1,t} > T_0, \cE_t \right \}  \right] \\
    &  = \rho_s v_s \cdot T_0 +
    \rho_s v_s \left (  \frac{2 \sqrt{\frac{2 \log T} {T_0}} - \frac{0.08\Delta_s}{\rho_s} \cdot \rho_1}{\rho_1 + 2 \sqrt{\frac{2 \log T} {T_0}}}  \right) \E  \left [\sum_{t = 1}^T   \mathbb{I} \left \{ A_t = 1, N_{1,t} > T_0, \cE_t \right \}  \right] \\
    & \le \rho_s v_s \cdot T_0 - 0.05 \Delta_s v_s \sum_{t = 1}^T   \Pr \left[A_t = 1, N_{1,t} > T_0, \tilde{\rho}_{s, t} \ge \rho_s + 0.08 \Delta_s, \tilde{\rho}_{1,t} \le \rho_1 + 2 \sqrt{\frac{3 \log(2nT)} {2N_{1, t}}}\right].
\end{align*}
It remains to bound the final expectation.
\begin{claim}
    \label{claim:rightwinner_sum_probs_1}
    $\sum_{t=1}^T \Pr[A_t \neq 1] \leq \frac{1}{nT} + \sum_{i=2}^n \frac{6 \log(2nT)}{\Delta_i^2}$.
\end{claim}
\begin{proof}
    Note that $\sum_{t=1}^T \Pr[A_t \neq 1] = \sum_{t=1}^T \E[\ind{A_t \neq 1}] \leq \sum_{i=2}^n \E[N_{i, T}]$.
    By Claim~\ref{lem:bound_Nit}, we have that $\E[N_{i, T}] \leq T \cdot \frac{1}{n^2T^2} + \frac{6 \log(2nT)}{\Delta_i^2}$.
    Taking the sum proves the claim.
\end{proof}
\begin{claim}
    \label{claim:rightwinner_sum_probs_2}
    $\sum_{t=1}^T \Pr\left[N_{1, t} \leq T_0 \right] \leq \sum_{i=2}^T \frac{6 \log(2nT)}{\Delta_i^2} + T_0 + \frac{1}{nT}$.
\end{claim}
\begin{proof}
    Let $t > \sum_{i=2}^n \frac{6 \log(2nT)}{\Delta_i^2} + T_0$.
    We have that
    \[
        \ind{N_{1,t} \leq T_0} = \ind{\sum_{i=2}^T N_{i,t} > t - T_0} \leq \sum_{i=2}^n \ind{N_{i,t} > \frac{6 \log(2nT)}{\Delta_i^2}} \leq \sum_{i=2}^n \ind{N_{i, T} > \frac{6 \log(2nT)}{\Delta_i^2}}.
    \]
    Taking expectations and applying Lemma~\ref{lem:bound_Nit} gives that $\Pr\left[ N_{i,t} \leq T_0 \right] \leq \frac{1}{nT^2}$.
    The claim follows by summing over all $t$ and using the trivial inequality $\ind{N_{1, t} \leq T_0} \leq 1$ for $t \leq \sum_{i=2}^n \frac{6 \log(2nT)}{\Delta_i^2} + T_0$.
\end{proof}
\begin{claim}
    \label{claim:rightwinner_sum_probs_3}
    $\sum_{t=1}^T \Pr\left[ \tilde{\rho}_{1,t} > \rho_1 + 2\sqrt{\frac{3 \log(2nT)}{2N_{1,t}}} \right] \leq \frac{1}{n^2T}$.
\end{claim}
\begin{proof}
    Follows directly from Lemma~\ref{lem:ucb_estimator}.
\end{proof}
\begin{claim}
    \label{claim:rightwinner_sum_probs}
    We have that
    \begin{align*}
        \sum_{t = 1}^T   & \Pr \left[A_t = 1, N_{1,t} > T_0, \tilde{\rho}_{s, t} \ge \rho_s + 0.08 \Delta_s, \tilde{\rho}_{1,t} \le \rho_1 + 2 \sqrt{\frac{3 \log(2nT)} {2N_{1, t}}}\right] \\
        & \geq T - \left( \frac{5}{nT} + T_0 + \sum_{i=2}^n \frac{12 \log(2nT)}{\Delta_i^2} \right).
    \end{align*}
\end{claim}
\begin{proof}
    Combining Claim~\ref{claim:rightwinner_sum_probs_1}, Claim~\ref{claim:rightwinner_sum_probs_2}, Claim~\ref{claim:Delta_gap}, and Claim~\ref{claim:rightwinner_sum_probs_3}, we have
    \begin{align*}
        \sum_{t=1}^T & \Pr \left[A_t = 1, N_{1,t} > T_0, \tilde{\rho}_{s, t} \ge \rho_s + 0.08 \Delta_s, \tilde{\rho}_{1,t} \le \rho_1 + 2 \sqrt{\frac{3 \log(2nT)} {2N_{1, t}}}\right] \\
        & \geq T - \sum_{t=1}^T \Pr[A_t \neq 1] + \Pr[N_{1,t} \leq T_0] + \Pr[\tilde{\rho}_{s,t} < \rho_s + 0.08 \Delta_s] + \Pr\left[\tilde{\rho}_{1,t} > \rho_1 + 2\sqrt{\frac{3\log(2nT)}{2N_{1,t}}} \right] \\
        & \geq T - \left( \frac{5}{nT} + T_0 + \sum_{i=2}^n \frac{12 \log(2nT)}{\Delta_i^2} \right),
    \end{align*}
     as desired.
\end{proof}
The claim now follows by combining the previous three claims.
\end{proof}

\section{Missing Proofs from Section~\ref{sec:global_ic_static}}
\label{app:global_ic_static}
\begin{proof}[Proof of Lemma~\ref{lem:static_exploration}]
Let $\hat{\rho}_{i, k}$ be the empirical estimate of $\rho_i$ after we show ad $i$ exactly $k$ times.
Let $\tilde{\rho}_{i, k} = \hat{\rho}_{i, k} + \sqrt{\frac{3 \log T}{2 k}}$ (resp.~$L_{i,k} = \hat{\rho}_{i,k} - \sqrt{\frac{3 \log T}{2k}}$) be the UCB (resp.~LCB) estimate after showing ad $i$ exactly $k$ times.
First we show that (i) $\tilde{\rho}_{1,k} v_1 > L_{j, k} v_j$ for all $j \in \{2, \ldots, n\}$ and $k \in [T]$ with high probability.
Thus, we never declare any $j \neq 1$ as the clear winner.
Next, we show that (ii) for some $K = O(\log T)$, we have $L_{1, K} v_1 > \tilde{\rho}_{j,k} v_j$ for $j \neq 1$ so that the exploration phase ends by the time we pull each arm $K$ times.

To prove (i), a straightforward application of Hoeffding's Inequality (Lemma~\ref{lem:hoeffding}) gives that with probability $1 - 2n^2/T$, we have $|\hat{\rho}_{i,k} - \rho_i| \leq \sqrt{\frac{\log(T)}{k}}$ for all $k \in [T]$.
We condition on this event.
Thus, $\hat{\rho}_{i,k} > \rho_{i,k} + 0.2 \sqrt{\frac{\log(T)}{k}}$ and $L_{i,k} < \rho_{i,k} - 0.2 \sqrt{\frac{\log(T)}{k}}$.
We conclude that $\tilde{\rho}_{1,k} v_1 > \rho_{1,k} v_1 > \rho_{j,k} v_j > L_{j, k} v_j$ which proves (i).

We now prove (ii). Taking $K \geq O\left(\frac{\log(T)}{\Delta_j^2} \right) = O(\log T)$, some straightforward calculations give that
\begin{align*}
    L_{1,K} v_1
    & > \left( \rho_1 - 0.2\sqrt{\frac{\log(T)}{K}} \right) v_1 \\
    & \geq \left( \rho_j + 0.2\sqrt{\frac{\log(T)}{K}} \right) v_j \\
    & > \tilde{\rho}_{j, K} v_j,
\end{align*}
which proves (ii).
\end{proof}

\end{document}